%% file: Arxiv_main.tex
\documentclass{article}

\usepackage[english]{babel}

\usepackage[letterpaper,top=2cm,bottom=2cm,left=3cm,right=3cm,marginparwidth=1.75cm]{geometry}

\usepackage{amsmath}
\usepackage{amsfonts}
\usepackage{amsthm}
\usepackage{amssymb}
\usepackage{graphicx}
\usepackage[dvipsnames]{xcolor}
\usepackage{booktabs}
\usepackage{subcaption}
\usepackage{tikz,pgfplots}
\usetikzlibrary{positioning}
\pgfplotsset{compat=1.18}
\usepackage[hidelinks]{hyperref}
\usepackage[normalem]{ulem}

\newtheorem{theorem}{Theorem}
\newtheorem{lemma}{Lemma}
\newtheorem{definition}{Definition}
\newtheorem{corollary}{Corollary}

\newtheorem{remark}{Remark}

\DeclareMathOperator*{\argmax}{arg\,max}
\DeclareMathOperator*{\argmin}{arg\,min}

\makeatletter
\def\@fnsymbol#1{\ifcase#1\or 1\or 2\or 3\or 4\or 5\or 6\else\@ctrerr\fi}
\makeatother

\title{Equilibrium and Infeasibility: A new solution concept for games}

\author{Anne Reulke\textsuperscript{1,2}, Mikaël Touati\textsuperscript{1}, Rachid El-Azouzi\textsuperscript{2}}

\begin{document}

\maketitle

\footnotetext[1]{Orange Research, 46 Av. de la République, 92320 Châtillon, France 
        {\small (anne.reulke@orange.com, mikael.touati@orange.com)}}
\footnotetext[2]{Laboratoire d'Informatique d'Avignon (LIA), 339 chemin des Meinajariés, 84000 Avignon, France 
        {\small (rachid.elazouzi@univ-avignon.fr)}}

\begin{abstract}
\input{Abstract}
\end{abstract}

\input{Introduction_new_short}
\input{Setting}
\input{Main_Result_no_equivalence}
\input{Result_bound_on_u_and_psi}
\input{Example}
\input{Relation_to_others}
\input{Conclusion_short}

\bibliographystyle{alpha}
\bibliography{Ref}

\appendix
\input{Appendix_short}

\end{document}

%% file: Abstract.tex
Addressing infeasibility in non-cooperative games has become an important topic, as many problems across different applications face this issue. We propose a new solution concept for generalized games with possibly infeasible individual constraints. A solution is defined as the limit of a sequence of generalized Nash equilibria induced by games with penalty terms relaxing the individual constraints. Existence is established for a broad range of games and 
we provide conditions allowing to characterize a $\psi$-penalized solution as a strategy profile maximizing every player's utility over all her penalty minimizing strategies. A variation of Divide-the-Dollar serves as an illustrative example.

%% file: Introduction_new_short.tex
\section{Introduction} 
Originally introduced to model a competitive economy \cite{debreu1952social, arrow1954existence}, generalized games have also been used to study various applications (e.g., \cite{azouzi2003, LeCadre2020, braouezec2023banking}). These games are usually solved by using the \textit{generalized Nash equilibrium} (GNE) \cite{debreu1952social} defined as a strategy profile that is feasible for all players and without any feasible benefiting deviation. We refer to \cite{FacchineiKanzow2010} for an extensive survey. All standard existence theorems \cite{Dutang2013} assume that the constraint correspondences are nonempty-valued. Under this assumption, a player has always at least one feasible strategy. Moreover, it is reasonable to assume that she prefers feasibility implying that infeasible profiles are not stable and that the GNE is an appropriate solution concept. Typically, constraint correspondences take only nonempty values in case of a \textit{coupled constraint} \cite{rosen1965}, i.e. the constraint correspondences are induced by a common set for all players, but this may not hold for general \textit{individual constraints}. Nevertheless, \cite{braouezec2023} stated the importance of providing solutions to games with constraint correspondences taking also empty values, particularly those without a GNE. To address this problem, \cite{braouezec2023} considers the set of jointly feasible strategy profiles induced by the individual constraints and solves the corresponding game with coupled constraint. However, this approach relies on the assumption that there is a jointly feasible strategy profile and includes only feasible strategy profiles as solutions. It remains thus an open question if strategy profiles that are infeasible for some players could also be considered credible. \\
The algorithm proposed in \cite{FacchineiKanzow2010penaltymethods}, based on a penalization of the constraints, computes an accumulation point that is a GNE or a possibly infeasible generalized stationary point. It can also be used for games with constraint correspondences with empty values or no jointly feasible strategy profile. Nevertheless, it was only applied to games with a GNE and the possibility and consequences of defining all accumulations points as a new solution concept has not been discussed.
\\
Essentially, each player is assumed to maximize her utility under her individual constraints and if the set of feasible strategies is empty, she faces an infeasible optimization problem. Solving such problems has also received little attention. An approach \cite{Chinneck2008}, consists in detecting infeasibility and proposing a feasible model. Another approach is to design algorithms converging to a reasonable point even for infeasible problems (e.g., \cite{Byrd2010,Chiche2010, dai2023, andrews2025}).
\\
In this paper, we introduce a new solution concept for generalized games, including those with constraint correspondences taking empty values, without a jointly feasible strategy profile or without a GNE. We relax the individual constraints and incorporate them into the utility function by a penalty function $\psi$ \cite{courant1943variational}. Iteratively increasing the penalty induces a sequence of \textit{$\psi$-penalized games} and corresponding GNEs, a \textit{$\psi$-penalized solution} being defined as their limit. The formal definition is given in Section~\ref{Sec: setting} and our main results in Section~\ref{Sec: statement}.
Finally, we show an illustrative example in Section~\ref{Sec: DD_with_constraints} and establish the relation to the generalized Nash equilibrium and the \textit{Nash bargaining solution} in Section~\ref{Sec: relation}.
All notations used in this article can be found in Table~\ref{tab: notations}.

%% file: Setting.tex
\section{Model and Concept}
\label{Sec: setting}
We consider a N-player non-cooperative generalized game $G := (\mathcal N, (\mathcal A_i)_{i \in \mathcal N}, (u_i)_{i \in \mathcal N}, (C_i)_{i \in \mathcal N})$. The set of players is $\mathcal{N} = \{1, \ldots, N\}$ and the set of strategies available to each player $i \in \mathcal N$ is $\mathcal A_i$. Following standard notation, we write $\mathcal{A}$ for the product $\prod_{j \in \mathcal N} \mathcal A_j$ of all strategy sets and $\mathcal{A}_{-i}$ for the product $\prod_{j \in \mathcal N, j \neq i} \mathcal A_j$ of all strategy sets except the one of player $i$. The utility function of player $i \in \mathcal N$ is $u_i: \mathcal A \to \mathbb R$. The set of feasible strategies for each player $i \in \mathcal N$ is given by a \textit{individual constraint correspondence} $C_i: \mathcal{A}_{-i} \to \mathcal{P}( \mathcal A_i)$ which associates to each strategy profile $a_{-i}\in \mathcal A_{-i}$ of the other players the set of feasible strategies $C_i(a_{-i}) \subset A_i$ for player $i$. We emphasize that the set of feasible strategies is thus dependent on the choice of strategies of the other players. Throughout this article, we assume that for any player $i \in \mathcal N$ the strategy set $\mathcal{A}_i$ is a convex and compact subset of $\mathbb R$. 
Moreover, we assume that for any player $i\in\mathcal N$ the utility function $u_i$ is continuous and concave in $a_i$, i.e., for all $a_{-i} \in \mathcal A_{-i}$, $u_i(\cdot, a_{-i})$ is concave. 
Additionally, we consider a convex and compact \textit{coupled constraint} $\mathcal L \subseteq \mathcal A$ that is shared by all players and induces for each player $i \in \mathcal N$ the \textit{coupled constraint correspondence} $L_i: \mathcal A_{-i} \to \mathcal P(\mathcal A_i)$ such that $L_i(a_{-i}) := \big\{ a_i \in \mathcal A_i ~|~ (a_i, a_{-i}) \in \mathcal L \big\}$. 
Without loss of generality, we assume that $\bigcup_{a_{-i} \in \mathcal A_{-i}} L_i(a_{-i}) = \mathcal A_i$ (see \cite{rosen1965}) and that for any player $i \in \mathcal N$ and strategy profile $a_{-i} \in \mathcal A_{-i}$, $C_i(a_{-i}) \subseteq L_i(a_{-i})$. 
\\
The aim of each player $i \in \mathcal N$ is to maximize her utility function $u_i: \mathcal A \to \mathbb R$ while respecting her constraints, i.e. to choose a strategy $a_i \in \mathcal A_i$, given the other players' strategies $a_{-i}$, that solves the maximization problem 
\begin{equation}
\label{Equ: individual_max_u}
    \begin{aligned}
        &\max_{a_i \in \mathcal A_i} u_i(a_i, a_{-i}) \\
        &\text{s.t. } a_i \in C_i(a_{-i}).
    \end{aligned}
\end{equation}
We emphasize that we do not assume that the individual constraints are nonempty-valued and problem~\eqref{Equ: individual_max_u} may be infeasible for some strategy profile $a_{-i} \in \mathcal A_{-i}$. We propose therefore to relax the individual constraints and incorporate a penalty function into the utility. 

\subsection{\texorpdfstring{$\psi$}{psi}-penalized game}
\label{sec: penalized_game}
We introduce the N-player generalized \textit{$\psi$-penalized game} $G_\rho := (\mathcal N, (\mathcal A_i)_{i \in \mathcal N}, (u_i - \rho \psi_i)_{i \in \mathcal N}, (L_i)_{i \in \mathcal N})$, an auxiliary parametrized game with coupled constraints and penalized utilities induced by game $G$. As before, the set of players is $\mathcal N$, the set of strategies of player $i \in \mathcal N$ is $\mathcal A_i$ but the only constraints are given by the coupled constraint correspondences $L_i$. 
The individual constraints are relaxed and replaced by a penalty term that consists of a penalty parameter $\rho > 0$ and a penalty function $\psi_i: \mathcal A \to \mathbb R_{\geq 0}$ such that for any $i \in \mathcal N$ and $a \in \mathcal A$
\begin{equation}
\label{equ: conditions_on_psi}
    a_i \in  \overline{C_i(a_{-i})} \Leftrightarrow \psi_i(a) = 0
\end{equation}
where $\overline{C_i(a_{-i})}$ is the closure of $C_i(a_{-i})$.
Furthermore, we assume that $\psi_i$ is continuous and convex in $a_i$, i.e., for any $a_{-i} \in \mathcal A_{-i}$, $\psi_i(\cdot, a_{-i})$ is convex. The utility function of player $i$ in $G_\rho$ is defined as the sum of the original utility and a penalty term, i.e. $u_i - \rho \psi_i$, such that, for any $a_{-i}\in \mathcal A_{-i}$, player $i$ solves the maximization problem 
\begin{equation*}
    \begin{aligned}
        &\max_{a_i \in \mathcal A_i} u_i(a_i, a_{-i}) - \rho \psi_i(a_i, a_{-i})\\
        &\text{s.t. } a_i \in L_i(a_{-i}).
    \end{aligned}
\end{equation*}
We only relaxed the individual constraints $C_i$, while the coupled constraint $\mathcal L$ is left unchanged. As discussed earlier, games with only a coupled constraint $\mathcal L$ have been extensively studied \cite{rosen1965, Dutang2013} and do not bear the difficulty of infeasibility. We follow hereby the idea of a partial penalization presented in \cite{FacchineiLampariello2011} but use it to define a new solution concept, not to solve a GNE problem.

\subsection{\texorpdfstring{$\psi$}{psi}-penalized solution concept}
Let $(\rho_n)_{n \in \mathbb N} \subseteq \mathbb R_{>0}$ be any sequence such that $\rho_n \to \infty$ and $(G_{\rho_n})_{n \in \mathbb N}$ the corresponding sequence of $\psi$-penalized games (as defined in Section \ref{sec: penalized_game}). In the following, we show that our assumptions imply that each game $G_{\rho_n}$ has at least one GNE $a_n^*\in\mathcal A$ (see Theorem \ref{Th: Existence}) and define a \textit{$\psi$-penalized solution} $a^*$ of game $G$ as an accumulation point (see Definition 1.1, p.~209 in \cite{Dugundji1978}) of any such sequence $(a_n^*)_{n \in \mathbb N}$.
We now give the formal definition of a $\psi$-penalized solution of game $G$.
\begin{definition}
\label{Def: penalized_solution}
    A strategy profile $a^*$ is called a \textit{$\psi$-penalized solution} of game $G$ \textit{under the constraint $\mathcal L$} if there is a sequence $(a_n^*)_{n \in \mathbb N}$ such that, for all $n \in \mathbb N$, $a_n^*$ is a GNE of the $\psi$-penalized game $G_{\rho_n}$ and $a^*$ is an accumulation point of $(a_n^*)_{n \in \mathbb N}$.
\end{definition}
A $\psi$-penalized solution may be infeasible for certain players and does not, in contrast to the GNE, demand feasibility for all players. It provides thus a solution to generalized games with constraint correspondences taking empty values or defining disjoint feasible sets. We show that a $\psi$-penalized solution exists under the initially stated assumptions. 
\begin{theorem}
\label{Th: Existence}
    There exists a $\psi$-penalized solution of game~$G$.
\end{theorem}
\begin{proof}
    For all players $i \in \mathcal N$ and all $\rho \in \mathbb R_{>0}$, the utility function $u_i - \rho \psi_i$ of the $\psi$-penalized game $G_\rho$ is continuous as the composition of two continuous functions $u_i$ and $\psi_i$. Moreover, for all $a_{-i} \in \mathcal{A}_{-i}$, $u_i(\cdot, a_{-i}) - \rho \psi_i(\cdot, a_{-i})$ is concave as the sum of the two concave functions $u_i(\cdot, a_{-i})$ and $-\rho \psi_i(\cdot, a_{-i})$. Hence, for all $n \in \mathbb N$, the $\psi$-penalized game $G_{\rho_n}$ is a concave game under the coupled convex and compact constraint $\mathcal L$ and has a GNE $a^*_n$ (see Theorem 1 in \cite{rosen1965}). We can thus construct a sequence $(a_n^*)_{n \in \mathbb N}$ such that $a_n^*$ is a GNE of $\psi$-penalized game $G_{\rho_n}$. Since $\mathcal L$ is compact, we know that $(a_n^*)_{n \in \mathbb N}$ has an accumulation point (Bolzano–Weierstrass theorem, see 3.4.8 in \cite{Bartle2000}).
\end{proof}
We remark that, by definition, the $\psi$-penalized solution depends on $\psi := (\psi_i)_{i \in \mathcal N}$ and the associated sequence of $\psi$-penalized games $(G_{\rho_n})_{n \in \mathbb N}$. The dependence on exogenous parameters can be found in other solution concepts, as the \textit{normalized equilibrium} on the chosen Lagrange multipliers \cite{rosen1965} or the \textit{H-essential equilibrium} on the class of penalty functions \cite{vanDamme1991}. The properties of $\psi$-penalized solutions with respect to $\psi$, beyond the results shown in this paper, remain an open question (see Section \ref{Sec: Future_works}).

%% file: Main_Result_no_equivalence.tex
\section{Main Results}
\label{Sec: statement}
In this section, we show that a $\psi$-penalized solution can be characterized without relying on the sequence of $\psi$-penalized games $(G_{\rho_n})_{n \in \mathbb N}$. First, in Section~\ref{Sec: Penalized_is_Solution}, we provide conditions implying that a $\psi$-penalized solution $a^* \in \mathcal A$ solves for each player $i \in \mathcal N$
\begin{equation}
\label{Prob: Max_over_min}
    \begin{aligned}
        &\max_{a_i} u_i(a_i, a_{-i}^*) \\
        &~\text{s.t. } a_i \in \argmin_{b_i \in L_i(a_{-i}^*)} \psi_i(b_i, a_{-i}^*).
    \end{aligned}
\end{equation}
Thus, at a $\psi$-penalized solution a benefiting deviation in $u_i$ would result in an infeasible strategy with respect to the coupled constraint $\mathcal L$ or a higher penalty $\psi_i$. 
\\
In Section~\ref{Sec: counterexample}, we give a counterexample showing that in general it is not sufficient for a strategy profile to solve problem~\eqref{Prob: Max_over_min} for each player to be a $\psi$-penalized solution. Finally, in Section~\ref{Sec: sufficient_and_exact_penalization}, we provide sufficient (but not necessary) conditions.
\\
In the following, we have used the notion of \textit{lower semi-continuous} (l.s.c.) 
defined in \cite{berge1963}. We recall the definition in Appendix~\ref{Sec: Appendix}, where we also present a preliminary result for the following proof. 

\subsection{\texorpdfstring{$\psi$}{psi}-penalized solutions solving problem \texorpdfstring{\eqref{Prob: Max_over_min}}{(3)}}
\label{Sec: Penalized_is_Solution}
First, we show that a $\psi$-penalized solution of game $G$ minimizes the penalty function $\psi$ for every player in response to the strategies of the other players.
\begin{lemma}
    \label{Lem: solution_is_in_argmin_psi}
    Let $a^* \in \mathcal A$ be a $\psi$-penalized solution of game $G$. If for some player $i \in \mathcal N$ the coupled constraint correspondence $L_i$ is l.s.c., then $a_i^* \in \argmin_{L_i(a_{-i}^*)} \psi_i(\cdot, a_{-i}^*)$.
\end{lemma}
\begin{proof}
    By contradiction, we assume there is a strategy $a_i \in \mathcal A_i$ such that $a_i \in \mathcal L_i(a_{-i}^*)$ and 
    \begin{equation}
    \label{equ_in_proof: in_argmin_supposition}
        \psi_i(a_i^*, a_{-i}^*) - \psi_i(a_i, a_{-i}^*) > 0.
    \end{equation}
    $a^*$ is a $\psi$-penalized solution of game $G$, thus by definition there is  a sequence $(a_{n_k}^*)_{k \in \mathbb N} \subset \mathcal A$ such that, for all $k \in \mathbb N$, $a_{n_k}^*$ is a GNE of $G_{n_k} := G_{\rho_{n_k}}$ and converging to $a^*$. \\
    Moreover, since $L_i$ is l.s.c., and thus l.s.c. at $a_{-i}^*$, there is a sequence $(a_{n_k,i})_{k \in \mathbb N}$ such that $a_{n_k,i} \in L_i(a_{n_k,-i}^*)$ and $a_{n_k,i} \to a_i$ as $k \to \infty$ (Lemma~\ref{Lem: Hogan}). By our initial assumption \eqref{equ_in_proof: in_argmin_supposition} and the continuity of $\psi_i$, we deduce
    \begin{align*}
        0 < &\lim_{k \to \infty} \psi_i(a_{n_k,i}^*, a_{n_k,-i}^*)  - \psi_i(a_{n_k,i}, a_{n_k,-i}^*) \\
        = &\lim_{k \to \infty} \frac{1}{\rho_{n_k}} \Big[ \big(u_i(a_{n_k,i}, a_{n_k,-i}^*) - \rho_{n_k} \psi_i(a_{n_k,i}, a_{n_k,-i}^*) \big) \\
        &- \big(u_i(a_{n_k,i}^*, a_{n_k,-i}^*) - \rho_{n_k} \psi_i(a_{n_k,i}^*, a_{n_k,-i}^*) \big) \Big] \\
        &+ \frac{1}{\rho_{n_k}} \Big[ u_i(a_{n_k,i}^*, a_{n_k,-i}^*) - u_i(a_{n_k,i}, a_{n_k,-i}^*)  \Big].
    \end{align*}
    Moreover, since $a_{n_k}^*$ is a GNE of $G_{n_k}$ for all $k \in \mathbb N$, we have
    \begin{align*}
        \big(&u_i(a_{n_k,i}, a_{n_k,-i}^*) - \rho_{n_k} \psi_i(a_{n_k,i}, a_{n,-i}^*)\big) \\
        &- \big(u_i(a_{n_k,i}^*, a_{n_k,-i}^*) - \rho_{n_k} \psi_i(a_{n_k,i}^*, a_{n_k,-i}^*) \big) \leq 0.
    \end{align*}
    Furthermore, $u_i$ is continuous and $\mathcal L$ is compact, thus the difference $u_i(a_{n_k,i}^*, a_{n_k,-i}^*) - u_i(a_{n_k,i}, a_{n_k,-i}^*)$ is bounded by some constant $C>0$, i.e. for all $k \in \mathbb N$ $|u_i(a_{n_k,i}^*, a_{n_k,-i}^*) - u_i(a_{n_k,i}, a_{n_k,-i}^*)| \leq C$.
    We conclude
    \begin{align*}
        0 &< \lim_{k \to \infty} \frac{1}{\rho_{n_k}} \Big[ \big(u_i(a_{n_k,i}, a_{n_k,-i}^*) - \rho_{n_k} \psi_i(a_{n_k,i}, a_{n_k,-i}^*)\big) \\
        &\quad \quad - \big(u_i(a_{n_k,i}^*, a_{n_k,-i}^*) - \rho_{n_k} \psi_i(a_{n_k,i}^*, a_{n_k,-i}^*) \big) \Big] \\
        &\quad \quad + \frac{1}{\rho_{n_k}} \Big[ u_i(a_{n_k,i}^*, a_{n_k,-i}^*) - u_i(a_{n_k,i}, a_{n_k,-i}^*)  \Big] \\
        &\leq  0 + \lim_{k \to \infty} \frac{1}{\rho_{n_k}} \cdot C = ~0.
    \end{align*}
    This is a contradiction.
\end{proof}
Now, we show that under additional conditions on the regularity of the individual constraints, a $\psi$-penalized solution of game $G$ maximizes the original utility function among all minimizers of the penalty function.
\begin{theorem}
\label{Th: penalized_solves_problem}
    Assume that for player $i \in \mathcal N$ the following conditions hold
    \begin{itemize}
        \item[(i)] $L_i$ is l.s.c.,
        \item[(ii)] for all $a_{-i} \in \mathcal{A}_{-i}$, $C_i$ is l.s.c. in $a_{-i}$ or $C_i(a_{-i})$ is a singleton,
        \item[(iii)] for all $a_{-i} \in \mathcal A_{-i}$, $\psi_i(\cdot, a_{-i})$ is strictly convex on $L_i(a_{-i}) \setminus C_i(a_{-i})$.
    \end{itemize}
    If $a^*$ is $\psi$-penalized solution of game $G$ then $a^*_i$ solves problem~\eqref{Prob: Max_over_min} for player $i$. 
\end{theorem}
\begin{proof}
    We distinguish two cases, $C_i(a_{-i}^*) = \emptyset$ (i.e. player $i$ has no feasible strategy at $a_{-i}^*$) and $C_i(a_{-i}^*) \neq \emptyset$ (i.e. player $i$ has a feasible strategy at $a_{-i}^*$).
    \\
    \textit{First Case.} Let $C_i(a_{-i}^*) = \emptyset$. By Assumption (iii), $\psi_i(\cdot, a_{-i}^*)$ is strictly convex on $L_i(a_{-i}^*)$ implying that $\argmin_{L_i(a_{-i}^*)} \psi_i(\cdot, a_{-i}^*)$ is a singleton. 
    By Lemma~\ref{Lem: solution_is_in_argmin_psi}, $a_i^*$ is this unique solution implying that it also solves problem~\eqref{Prob: Max_over_min} as unique element in the feasible set. \\
    \textit{Second Case.} Let $C_i(a_{-i}^*) \neq \emptyset$. By construction of $\psi_i$ (see \eqref{equ: conditions_on_psi}) and Lemma~\ref{Lem: solution_is_in_argmin_psi}, $ a_i^* \in \argmin_{L_i(a_{-i}^*)} \psi_i(\cdot, a_{-i}^*) = \overline{C_i(a_{-i}^*)}$, thus we must show that for all $b_i \in \overline{C_i(a_{-i}^*)}$
    \begin{equation*}
        u_i(a_i^*, a_{-i}^*) \geq u_i(b_i, a_{-i}^*).
    \end{equation*}
    We consider the two cases induced by Assumption (ii).\\
    If $C_i(a_{-i}^*)$ is a singleton, then $\overline{C_i(a_{-i}^*)} = C_i(a_{-i}^*)$ and, as before, $a_{i}^*$ is the unique solution to problem~\eqref{Prob: Max_over_min}.\\
    If $C_i$ is l.s.c. at $a_{-i}^*$ then the correspondence $\overline{C}:\mathcal{A}_{-i}\rightarrow \mathcal P(\mathcal A_i)$ s.t. for all $a_{-i}\in\mathcal{A}_{-i}$ $\overline{C}(a_{-i}) := \overline{C_i(a_{-i})}$ is also l.s.c. (see Lemma \ref{Lem: compactification_is_lsc} in Appendix~\ref{Sec: Appendix}). \\    
    Besides, by definition of the $\psi$-penalized solution, there is a sequence $(a_{n_k}^*)_{k \in \mathbb N}$ converging to $a^*$ as $k \to \infty$ such that, for all $k \in \mathbb N$, $a_{n_k}^*$ is a GNE of $G_{n_k}:= G_{\rho_{n_k}}$. Furthermore, by Lemma~\ref{Lem: Hogan}, for all $b_i\in \overline{C_i}(a_{-i}^*) = \argmin_{L_i(a_{-i}^*)} \psi_i(\cdot, a_{-i}^*)$ there is a sequence $(b_{n_k,i})_{n \in \mathbb N}$ such that $b_{n_k,i} \in \overline{C_i}(a_{n_k,-i}^*)$ and $b_{n_k,i} \to b_i$ as $k \to \infty$.
    We have
    \begin{align*}
        u_i\big(&b_{n_k,i}, a_{n_k,-i}^*\big) \\
        &= u_i\big(b_{n_k,i}, a_{n_k,-i}^*\big) - \rho_{n_k} \psi_i\big(b_{n_k,i}, a_{n_k,-i}^*\big) \\ 
        &\leq u_i\big(a^*_{n_k,i}, a_{n_k,-i}^*\big) - \rho_{n_k} \psi_i\big(a^*_{n_k,i}, a_{n_k,-i}^*\big) \\
        &\leq u_i\big(a^*_{n_k,i}, a_{n_k,-i}^*\big).
    \end{align*}
    The first equality follows from the penalty vanishing at $(b_{n_k,i},a_{n_k,-i}^*)$ as $b_{n_k,i} \in \overline{C_i}(a_{n_k,-i}^*)$. The second line holds since $a^*_{n_k}$ is a GNE of $G_{n_k}$ and $b_{n_k,i} \in \overline{C_i}(a_{-i,n}^*) \subset L_i(a_{-i,n}^*)$. The third one follows by nonnegativity of $\psi_i$ and $\rho_n$.
    Finally, $u_i$ is continuous and the sequences $(a_{n_k}^*)_{k \in \mathbb N}$ and $(b_{n_k,i})_{n \in \mathbb N}$ are convergent, thus by passing to the limit we have $u_i\big(b_{i}, a_{-i}^*\big) \leq  u_i\big(a^*_{i}, a_{-i}^*\big)$, implying that $a_i^*$ solves problem~\eqref{Prob: Max_over_min}.
\end{proof}
\begin{remark}
    It is standard in the literature to assume that the constraint correspondences are l.s.c. and nonempty valued, i.e. for all $i\in\mathcal N$ and for all $a_{-i}\in\mathcal{A}_{-i}$, $C_i(a_{-i})\neq\emptyset$. While in Theorem~\ref{Th: penalized_solves_problem}, we also assume lower semi-continuity by Assumption (i), Assumption (ii) does not imply that the individual correspondences take nonempty values. In fact, by definition, a correspondence is l.s.c. in points with empty values but it is not at a transition point from empty to nonempty values, i.e. for any topological spaces $X, Y$ a correspondence $C: X \to Y$ is not l.s.c. in a point $x_0 \in X$ such that $C(x_0) \neq \emptyset$ but for all neighborhoods $U(x_0)$ of $x_0$ there is a $x \in U(x_0)$ such that $C(x) = \emptyset$. 
    To address this problem, by Assumption (ii), we demand that the correspondence associates a singleton to such transition points which guarantees, roughly speaking, a "smooth" transition from empty to nonempty values. Moreover, we remark that Assumption (ii) was imposed to guarantee the uniqueness of the minimizer of $\psi_i(\cdot, a_{-i})$ on $L_i(a_{-i}) \setminus C_i(a_{-i})$ and could therefore be replaced by assuming directly the uniqueness. \\
\end{remark}
\subsection{Solutions of problem\texorpdfstring{~\eqref{Prob: Max_over_min}}{(3)} may not be \texorpdfstring{$\psi$}{psi}-penalized solutions: a counterexample}
\label{Sec: counterexample}
\input{Counter_example_without_trains}

%% file: Counter_example_without_trains.tex
In this section, we provide a counterexample showing that, given penalty function $\psi$, it is not sufficient for a point to solve problem~\eqref{Prob: Max_over_min} for every player to be a $\psi$-penalized solution. Therefore, we consider a $2$-player game $G^\text{ce}$ with strategy sets $\mathcal A_i = [0,1]$ and utilities $u_i(a_i, a_{-i}) = - (1 - a_i)^2$. The coupled constraint is given by $\mathcal L = \{(a_1, a_2) \in [0,1]^2 ~|~ a_1 + a_2 \leq 1\}$ and the individual constraints of the players by
\begin{align*}
    C_1(a_2) &= \begin{cases}
        \{a_2\} &a_2 \in [0, 0.5] \\
        \emptyset &a_2 \in (0.5, 1]
    \end{cases} \\
\intertext{and}
    C_2(a_1) &= \begin{cases}
        [a_1, \min(3a_1, 1 - a_1)] &a_1 \in [0, 0.5] \\
        \emptyset &a_1 \in (0.5, 1].
    \end{cases}
\end{align*}
\\
We consider the penalty functions 
$\psi_1(a_1, a_2) = (a_1-a_2)^2$ and $\psi_2(a_1, a_2) = d((a_1,a_2), \Gamma(C_2))^2$ where $d$ is the distance to a set defined in Table~\ref{tab: notations} and $\Gamma(C_2)$ the graph of the correspondence $C_2$.
% \begin{align*}
%     \psi_1(a_1, a_2) &= (a_1-a_2)^2 \\
%     \psi_2(a_1, a_2) &= d((a_1,a_2), \Gamma(C_2))^2
% \end{align*}
The set of GNEs of the $\psi$-penalized game $G^\text{ce}_{\rho_n}$ is
\begin{equation*}
    %GNE(G_{\rho_n}) = 
    \bigg\{ (a_1, 1-a_1)\in [0,1]^2 ~|~ \frac{\rho_n}{4\rho_n + 10} \leq a_1 \leq \frac{\rho_n +1}{2 \rho_n + 1} \bigg\}.
\end{equation*}
% \begin{equation*}
%     GNE(G_{\rho_n}) = \Big\{ (1 - a_2, a_2) ~|~ \frac{\rho_n}{2 \rho_n + 1} \leq a_2 \leq \frac{3 \rho_n + 10}{4 \rho_n + 10} \Big\}
% \end{equation*}
Thus, the set of $\psi$-penalized solutions of game $G^\text{ce}$ (see Fig.~\ref{fig: Counterexample}) is
\begin{equation*}
    \bigg\{ (a_1, 1-a_1)\in [0,1]^2 ~|~ \frac{1}{4} \leq a_1 \leq \frac{1}{2} \bigg\}.
\end{equation*}
The point $(0,0)$ solves problem~\eqref{Prob: Max_over_min} for both players but is not a $\psi$-penalized solution of $G^\text{ce}$ implying that solving problem~\eqref{Prob: Max_over_min} for all players is not sufficient without further assumptions.
Moreover, $(0,0)$ is a GNE of $G^\text{ce}$, showing that a GNE of the original game may not be a $\psi$-penalized solution.

\begin{figure}[ht]
\centering
    \includegraphics[width=0.5\columnwidth]{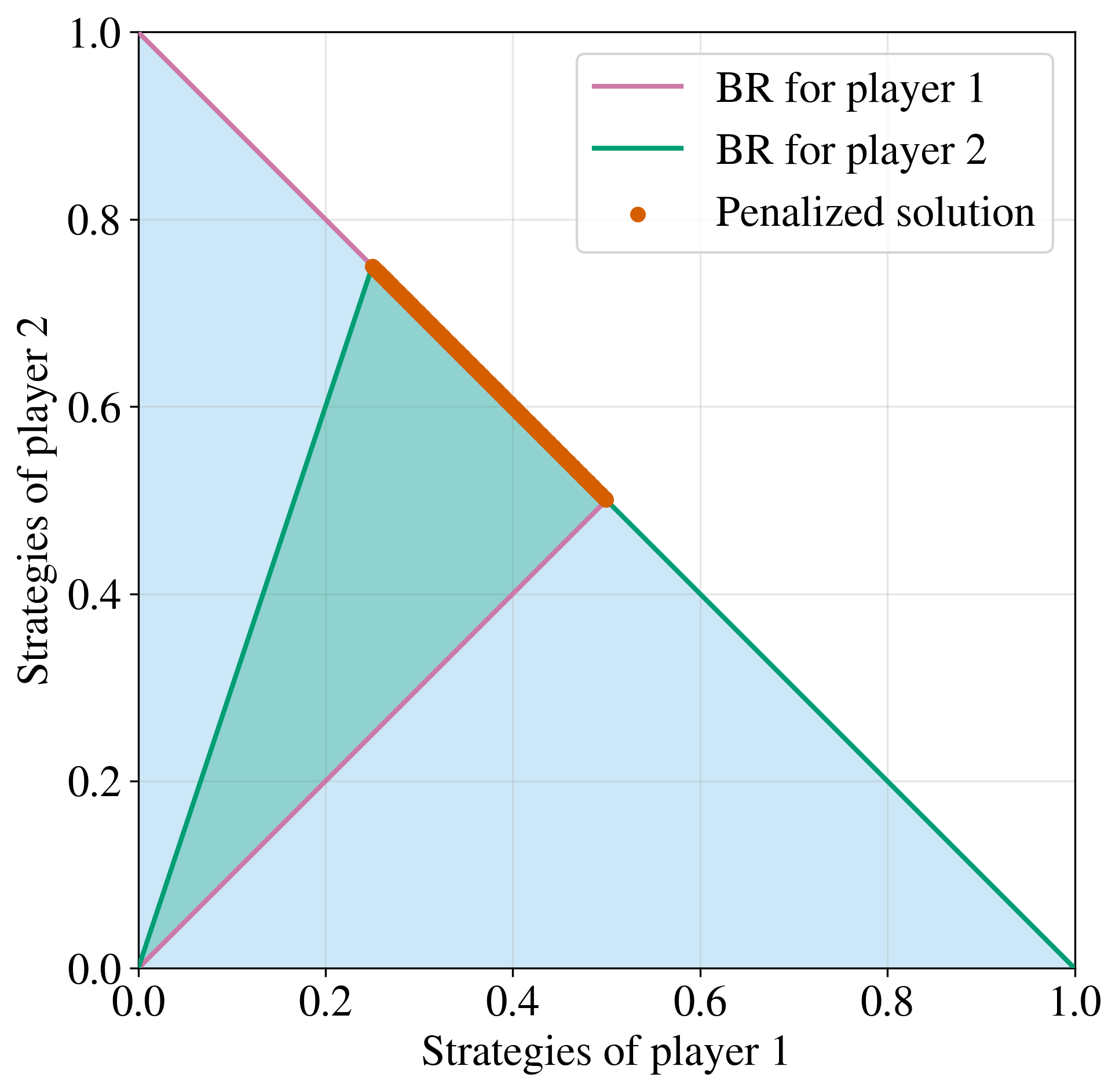}
\caption{$\psi$-penalized solutions (red) and best response correspondences for player 1 (purple) and player 2 (green) under the coupled constraint $\mathcal L$ (blue) and their respective individual constraints for game $G^\text{ce}$.}
\label{fig: Counterexample}
\end{figure}

%% file: Result_bound_on_u_and_psi.tex
\subsection{Sufficient conditions and exact penalization}
\label{Sec: sufficient_and_exact_penalization}
In this section, we give conditions under which every strategy profile solving problem~\eqref{Prob: Max_over_min} for all players is a $\psi$-penalized solution. 
\begin{theorem}
\label{Th: exact_penalization}
    Let $a^* \in L$. Assume there are constants $\alpha_u, \alpha_\psi>0$ such that for all players $i \in \mathcal N$ and strategies $a_i \in L_i(a_{-i}^*)$, $u_i(\cdot, a_{-i}^*)$ is $\alpha_u$-Lipschitz continuous at $a_i^*$, i.e.
    \begin{equation}\label{condition_on_u}
        |u_i(a_i, a_{-i}^*) - u_i(a^*)| \leq \alpha_u |a_i - a_i^*|
    \end{equation}
    and
    \begin{equation}\label{condition_on_psi}
        \begin{aligned}
            &| \psi_i(a_i, a_{-i}^*) - \psi_i(a^*) | \geq \alpha_{\psi} |a_i - a_i^*| \\
            \text{or } &\psi(a^*) = \psi(a_i, a_{-i}^*).
        \end{aligned}
    \end{equation}
    If $a^*$ solves problem~\eqref{Prob: Max_over_min} for all players, then $a^*$ is a $\psi$-penalized solution. 
\end{theorem}
\begin{proof}
    Let $a^*$ be a solution to problem~\ref{Prob: Max_over_min} for all players. 
    We show that $a^*$ is a GNE of every $\psi$-penalized game $G_{\rho_n}$ with $\rho_n > \alpha_u/\alpha_\psi$.
    Let $i \in \mathcal N$ and $a_i \in L_i(a_{-i}^*)$. 
    By definition of $a^*$ we have $\psi_i(a_i, a_{-i}^*) - \psi_i(a^*) \geq 0$.\\
    %%%%%%%%%%%%%%%%%%%%%%%%%%%%
    \textit{First case.} Let $\psi_i(a_i, a_{-i}^*) = \psi_i(a^*)$. Then, again by definition of $a^*$, we also have $u_i(a^*) \geq u_i(a_i, a_{-i}^*)$.
    Implying,
    \begin{equation*}
        u_i(a^*) - \rho_n \psi_i(a^*) \geq u_i(a_i, a_{-i}^*) - \rho_n \psi_i(a_i, a_{-i}^*). 
    \end{equation*}
    %%%%%%%%%%%%%%%%%%%%%%%%%%%%
    \textit{Second case.} Let $\psi_i(a_i, a_{-i}^*) >  \psi_i(a^*)$ and $ u_i(a^*) \geq u_i(a_i, a_{-i}^*)$. For all $\rho_n>0$ we have 
    \begin{equation*}
        u_i(a^*) - \rho_n \psi_i(a^*) > u_i(a_i, a_{-i}^*) - \rho_n \psi_i(a_i, a_{-i}^*).
    \end{equation*}
    %%%%%%%%%%%%%%%%%%%%%%%%%%%%
   \textit{Third case.} Let $\psi_i(a_i, a_{-i}^*) > \psi_i(a^*)$ and $u_i(a^*) < u_i(a_i, a_{-i}^*)$. For all $\rho_n > \alpha_u/\alpha_\psi$ we have
    \begin{align*}
        \rho_n &\Big( \psi_i(a_i, a_{-i}^*) - \psi_i(a^*) \Big) 
        >\frac{\alpha_u}{\alpha_\psi} \Big( \psi_i(a_i, a_{-i}^*) - \psi_i(a^*) \Big) \\
        & \geq \alpha_u |a_i - a_i^*| \geq |u_i(a_i, a_{-i}^*) - u_i(a^*)| \\
        &= u_i(a_i, a_{-i}^*) - u_i(a^*).
    \end{align*}
    where the second and third inequalities follow from \eqref{condition_on_psi} and \eqref{condition_on_u}.
    Thus, we conclude that $a^*$ is a GNE for all $\psi$-penalized games $G_{\rho_n}$ with $\rho_n > \alpha_u/\alpha_\psi$. 
    We can thus construct a sequence $(a_n)_{n \in \mathbb N}$ converging to $a^*$ such that $a_n$ is a GNE of $G_{\rho_n}$ by setting $a_n = a^*$ for all $n$ such that $\rho_n > \alpha_u/\alpha_\psi$. 
\end{proof}
\begin{remark}
\label{Rem; conditions_on_psi_u}
    Condition~\eqref{condition_on_psi} guarantees that the penalty is either constant or grows at least linearly. In the counterexample in~\ref{Sec: counterexample}, $\psi_1(\cdot, 0)$ does not satisfy Condition~\eqref{condition_on_psi} at $0$. In fact, $\psi_1(\cdot, 0)$ is differentiable at $0$ which minimizes $\psi_1(\cdot, 0)$ implying $\psi_1'(0,0) =0$ and preventing the existence of $\alpha_\psi >0$. 
\end{remark}
Moreover, we can show similarly that if for some $a^*\in\mathcal L$ the conditions hold then, for any pair $\rho_n, \rho_m > \alpha_u/\alpha_\psi$, $a^*$ is a GNE of $G_{\rho_n}$ if and only if $a^*$ is a GNE of $G_{\rho_m}$. If the conditions are satisfied at all strategy profiles $a \in \mathcal A$, then for all $\rho_n, \rho_m > \alpha_u/\alpha_\psi$ the set of GNEs is the same for $\psi$-penalized games $G_{\rho_n}$ and $G_{\rho_m}$ and the set of $\psi$-penalized solutions is the closure of the set of GNEs of  $G_{\rho_n}$ for $\rho_n > \alpha_u/\alpha_\psi$.  

%% file: Example.tex
\section{Example}
\label{Sec: DD_with_constraints}
We show now an illustrative example inspired by Divide-the-Dollar that allows an exact penalization.
Therefore, we consider a 2-player game $G^{DD}$ in which each player $i \in \{1, 2\}$ makes a demand $\mathcal A_i = [0,1]$. Their utility $u_i: \mathcal A \to \mathbb R$ is given by their respective demand, i.e. $u_i(a_i, a_{-i}) = a_i$. Both players must share a resource whose limited amount is modeled as a coupled constraint
\begin{equation*}
    \mathcal L = \Big\{ (a_1, a_2) \in [0,1]^2 ~|~ a_1 + a_2 \leq 1 \Big \}.
\end{equation*}
We assume that both players know the amount of the resource and, following the standard way to consider coupled constraints in generalized games, that player $i$'s feasible set of strategies is given by the constraint correspondences $L_i$. Besides, we assume that each player has additional personal constraints. First, each player $i$ requires at least $\delta_i > 0$. Second, a player means to leave unused an amount $\epsilon_i > 0$ of the resource. The individual constraint correspondence $C_i: [0,1] \to [0,1]$ is thus given by $C_i(a_{-i}) := [\delta_i,1 - \epsilon_i - a_{-i} ]$.
We consider the penalty function $\psi_i: [0,1]^2 \to [0, \infty)$ defined by
\begin{equation}
\label{eq:penalty_distance}
    \psi_i(a) = d\big(a_i, \Gamma(C_i)\big) := \inf_{b_i \in  \Gamma(C_i)} | a_i - b_i |.
\end{equation}
In this example, we consider the penalty function \eqref{eq:penalty_distance} as a natural metric to measure the distance of a point to a non-empty compact set (see the smoothing in \cite{Nash1953demand}). We remark that it is sufficient to consider the GNEs of $G_\rho$ with $\rho > 1/\sqrt{2}$, each point satisfying the conditions of Theorem~\ref{Th: exact_penalization} with $\alpha_u = 1$ and $\alpha_\psi = \sqrt{2}$.

In the following, we will distinguish two cases. First, we assume that both players intend to leave unused the same amount of the resource, i.e. $\epsilon_1 = \epsilon_2$. The individual constraints have a nonempty intersection in which all profiles are feasible for both players (see Fig.~\ref{fig: Solutions_example}a). 
Hence, the set of $\psi$-penalized solutions is given by the common upper right boundary of the individual constraints and coincides with the set of GNEs of game $G^{DD}$ (see Fig.~\ref{fig: Solutions_example}a). Second, we assume the players differ in how much they leave unused, i.e. $\epsilon_1 \neq \epsilon_2$. Without loss of generality, let $\epsilon_1 < \epsilon_2$. 
If the $\delta_1$ and $\delta_2$ are small enough, i.e. $\delta_1 + \delta_2 \leq 1 - \epsilon_1$, there is again a nonempty intersection of the individual constraints (see Fig.~\ref{fig: Solutions_example}b). All profiles in this intersection are feasible for both players. However, no profile in this intersection is maximizing the utility for player $1$ who can always benefit by deviating to a higher feasible demand. As a consequence, game $G^{DD}$ has no GNE. Nevertheless, the game has the $\psi$-penalized solution $(0.6, 0.3)$.
If the requirements $\delta_1$ and $\delta_2$ are large, i.e. $\delta_1 + \delta_2 > 1 - \epsilon_1$, the individual constraints are disjoint (see Fig.~\ref{fig: Solutions_example}c) and, by definition, there is no GNE of game $G^{DD}$. There is an infinity of $\psi$-penalized solutions on the upper right boundary of the coupled constraint which are all infeasible for both players.
\begin{figure}[ht]
\centering
\begin{minipage}[b]{0.3\columnwidth}
    \centering
    \includegraphics[width=\linewidth]{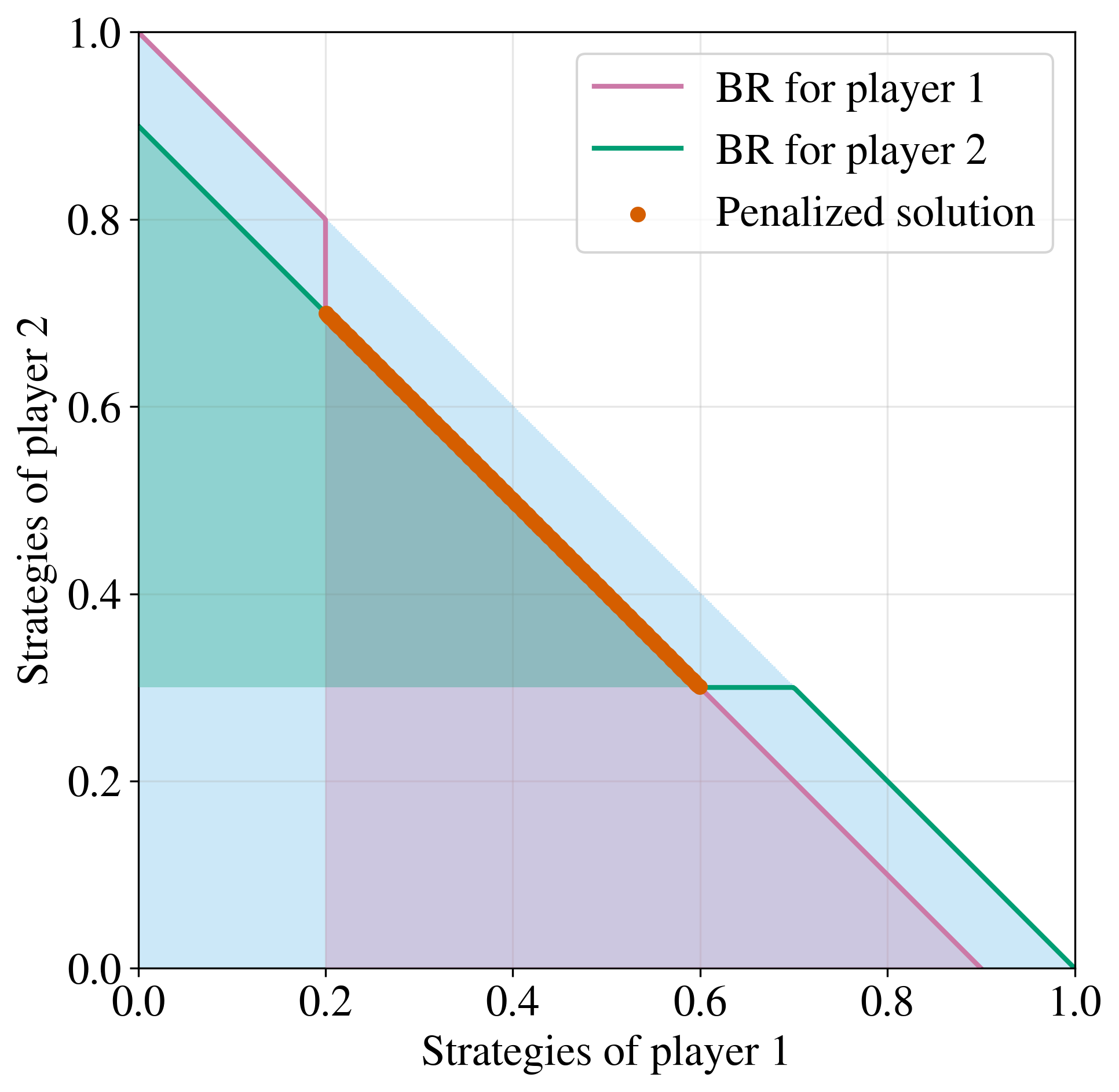}
    (a) $\delta_1 = 0.2$, $\delta_2 = 0.3$, $\epsilon_1 = \epsilon_2 = 0.1$
\end{minipage}
\hfill
\begin{minipage}[b]{0.3\columnwidth}
    \centering
    \includegraphics[width=\linewidth]{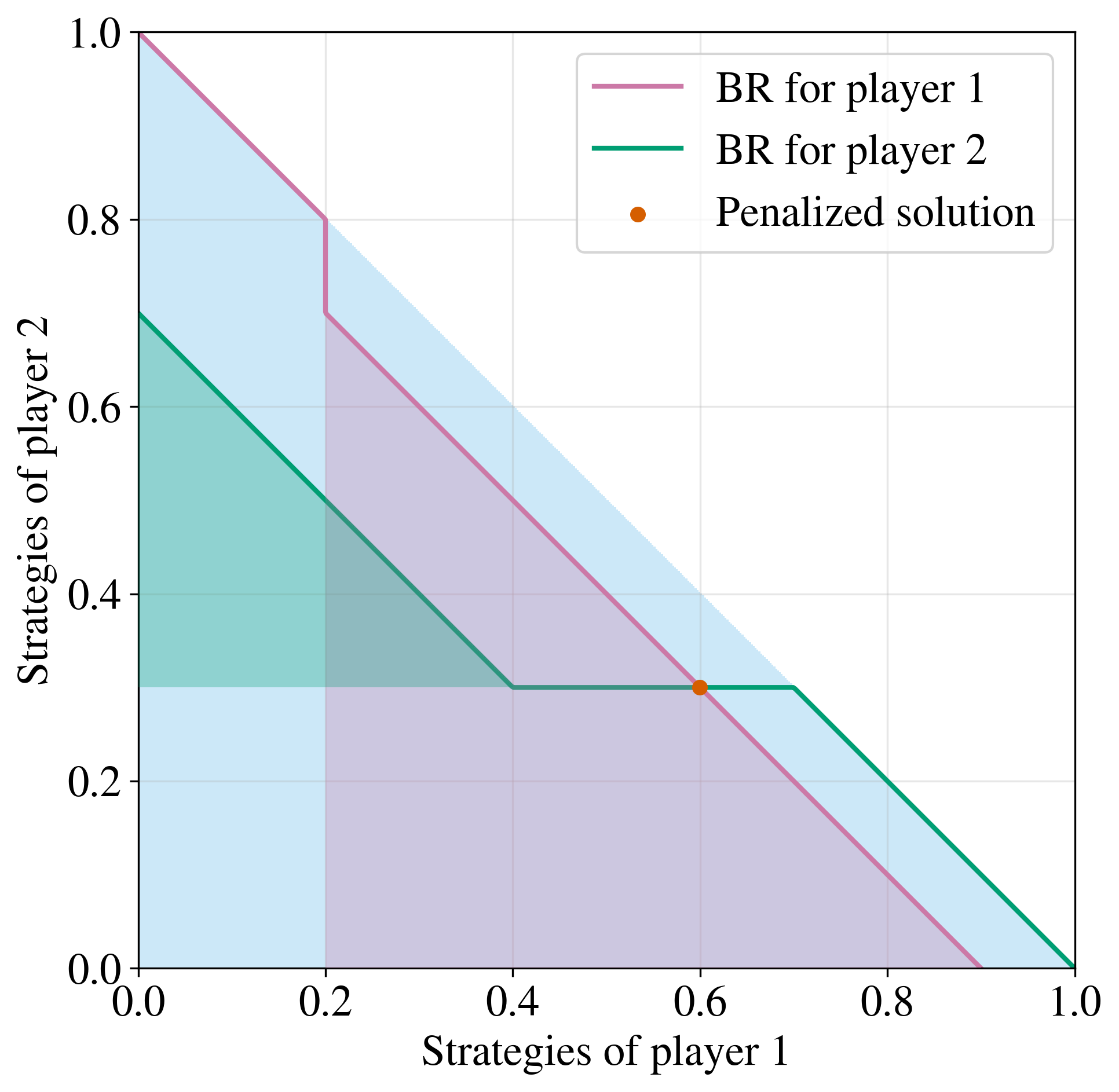}
    (b) $\delta_1 = 0.2$, $\delta_2 = 0.3$, $\epsilon_1 = 0.1$, $\epsilon_2 = 0.3$
\end{minipage}
\hfill
\begin{minipage}[b]{0.3\columnwidth}
    \centering
    \includegraphics[width=\linewidth]{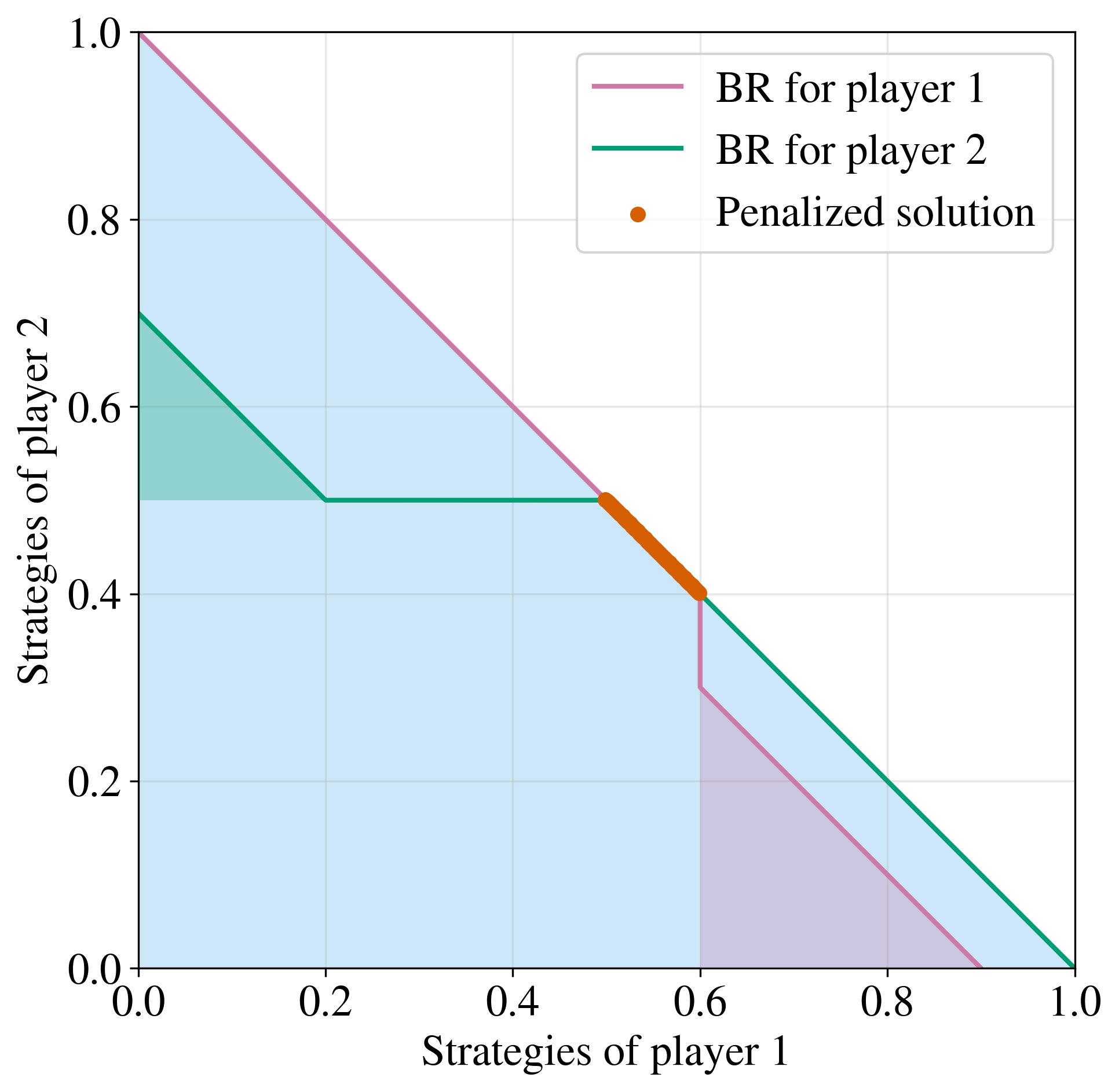}
    (c) $\delta_1 = 0.6$, $\delta_2 = 0.5$, $\epsilon_1 = 0.1$, $\epsilon_2 = 0.3$
\end{minipage}

\caption{$\psi$-penalized solutions (red) and best response correspondences for player 1 (purple) and player 2 (green) under the coupled constraint $\mathcal L$ (blue) and their respective individual constraints for game $G^\text{DD}$.}
\label{fig: Solutions_example}
\end{figure}

%% file: Relation_to_others.tex
\section{Relation to other concepts}
\label{Sec: relation}
In the following section, we investigate how our concept relates to the GNE and Nash bargaining solution. While the GNE is the standard solution concept for generalized games, Nash's solution to the bargaining problem shows how the players may attain agreement by smoothing the original game. 

\subsection{Generalized Nash equilibrium}
First, we remark that if the individual constraints are induced by the coupled constraint $\mathcal L$, i.e. for all players $i \in \mathcal N$ $C_i := L_i$, the penalty function is null by definition. In this case, the sequence of $\psi$-penalized games is constant and identically equal to $G$ implying that the set of $\psi$-penalized solution is the closure of the set of GNEs of $G$. \\
Considering the general case, the following corollary shows that every $\psi$-penalized solution of game $G$ that is feasible for all players is a GNE.
\begin{corollary}
\label{Cor: penalized_is_GNE_less_cond}
    Let $a^*$ be a $\psi$-penalized solution of game $G$. If for all players $i \in \mathcal N$, $C_i$ is l.s.c. and $a_i^* \in C_i(a_{-i}^*)$, then $a^*$ is a GNE of game $G$.
\end{corollary}
\begin{proof}
    See second case in proof of Theorem \ref{Th: penalized_solves_problem}.
\end{proof}
The individual feasibility for all players is necessary by definition of a GNE but the lower semicontinuity of the individual constraints is not.
In fact, as an example, consider a variation of Example \ref{fig: Solutions_example}a obtained by replacing the individual constraint correspondence $C_2$ by $C_2':\mathcal{A}_{1}\rightarrow \mathcal{P}(\mathcal{A}_2)$ such that
\begin{align*}
    C_2'(a_1) = 
        \begin{cases}
            [0.3, 0.6] &a_1 \in [0,0.1)\\
            C_2(a_1) &a_1 \in [0,1, 1]
        \end{cases} 
\end{align*}
This correspondence is not l.s.c. at $0.1$ but the set of $\psi$-penalized solutions of game $G$ is the same as in Example \ref{fig: Solutions_example} and continues to coincide with the set of GNEs of game $G$ (see Fig.~\ref{fig: GNE_counterexamples}a). The conditions stated in Corollary~\ref{Cor: penalized_is_GNE_less_cond} are thus sufficient but not necessary. \\
Nevertheless, assuming only that a $\psi$-penalized solution of $G$ is feasible for all players is not sufficient to guarantee that it is a GNE of $G$. As an example, consider a second variation of Example \ref{fig: Solutions_example}a with $C_2'':\mathcal{A}_{1}\rightarrow \mathcal{P}(\mathcal{A}_2)$ such that
\begin{align*}
    C_2''(a_1) = \begin{cases}
        [0.3, 0.95-a_1] &a_1 \in [0,0.4] \\
        C_2(a_1)  &a_1 \in (0.4, 1]
    \end{cases}
\end{align*}
This correspondence is not l.s.c. at $0.4$. Nevertheless, if the penalty parameter $\rho > 1/\sqrt{2}$, it can be shown that the set of GNEs of the $\psi$-penalized games $G_\rho$ does not depend on $\rho$ and is given by
\begin{equation*}
    \Big\{ (a, 0.9 - a) ~|~ a  \in (0.4, 0.6] \Big\}
\end{equation*}
Therefore, there is a sequence of GNEs in the associated $\psi$-penalized games converging to $(0.4, 0.5)$ which is thus a $\psi$-penalized solution of game $G$. The profile $(0.4, 0.5)$ is feasible for both players but player $2$ has a feasible benefiting deviation to $0.55$ implying that $(0.4, 0.5)$ is not a GNE.
\\
\begin{figure}[ht]
\centering
    \begin{minipage}[b]{0.4\columnwidth}
        \centering
        \includegraphics[width=\linewidth]{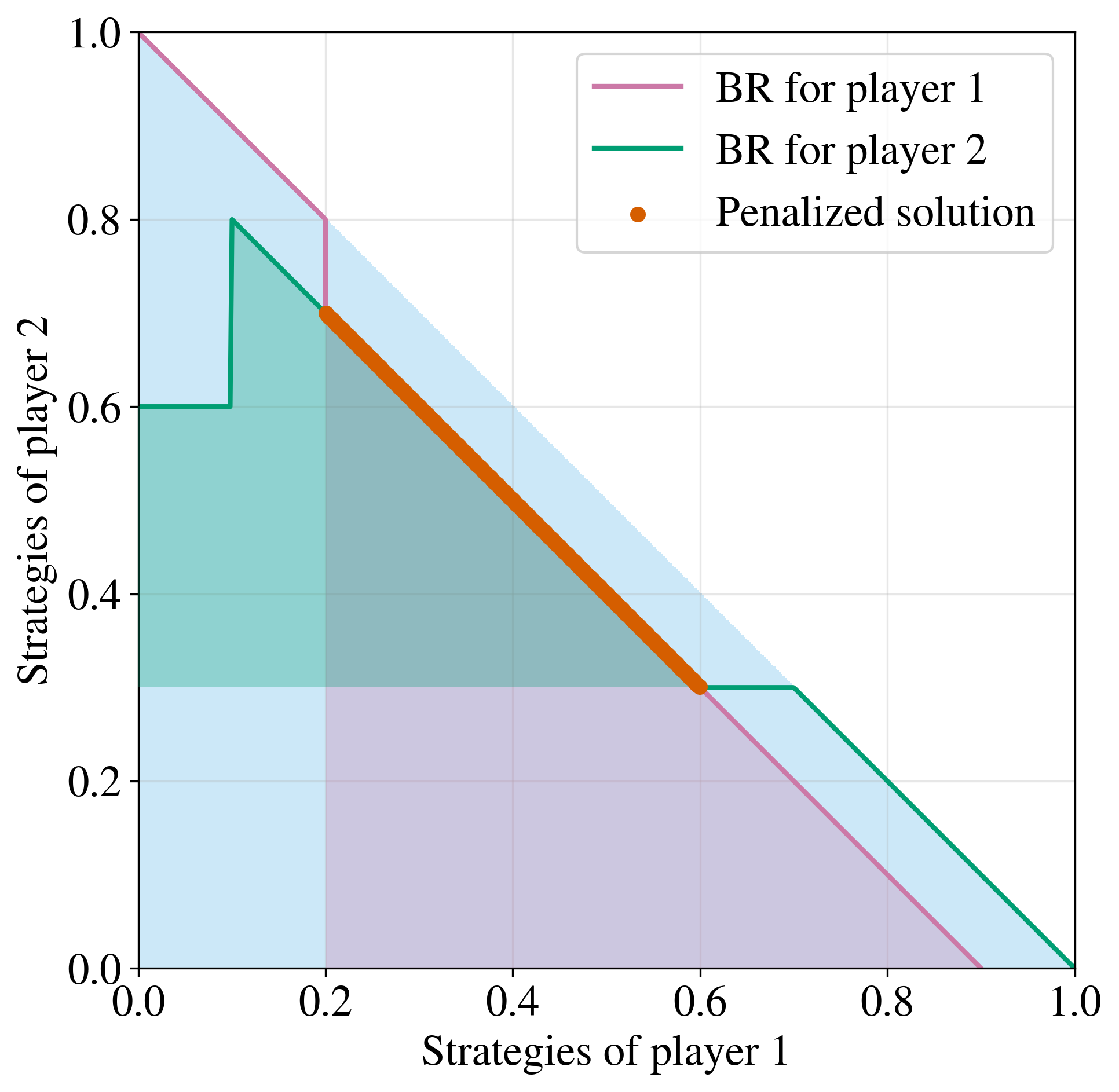}
        (a) with $C_2'$
    \end{minipage}
    %\hfill
    \hspace{1cm}
    \begin{minipage}[b]{0.4\columnwidth}
        \centering
        \includegraphics[width = \linewidth]{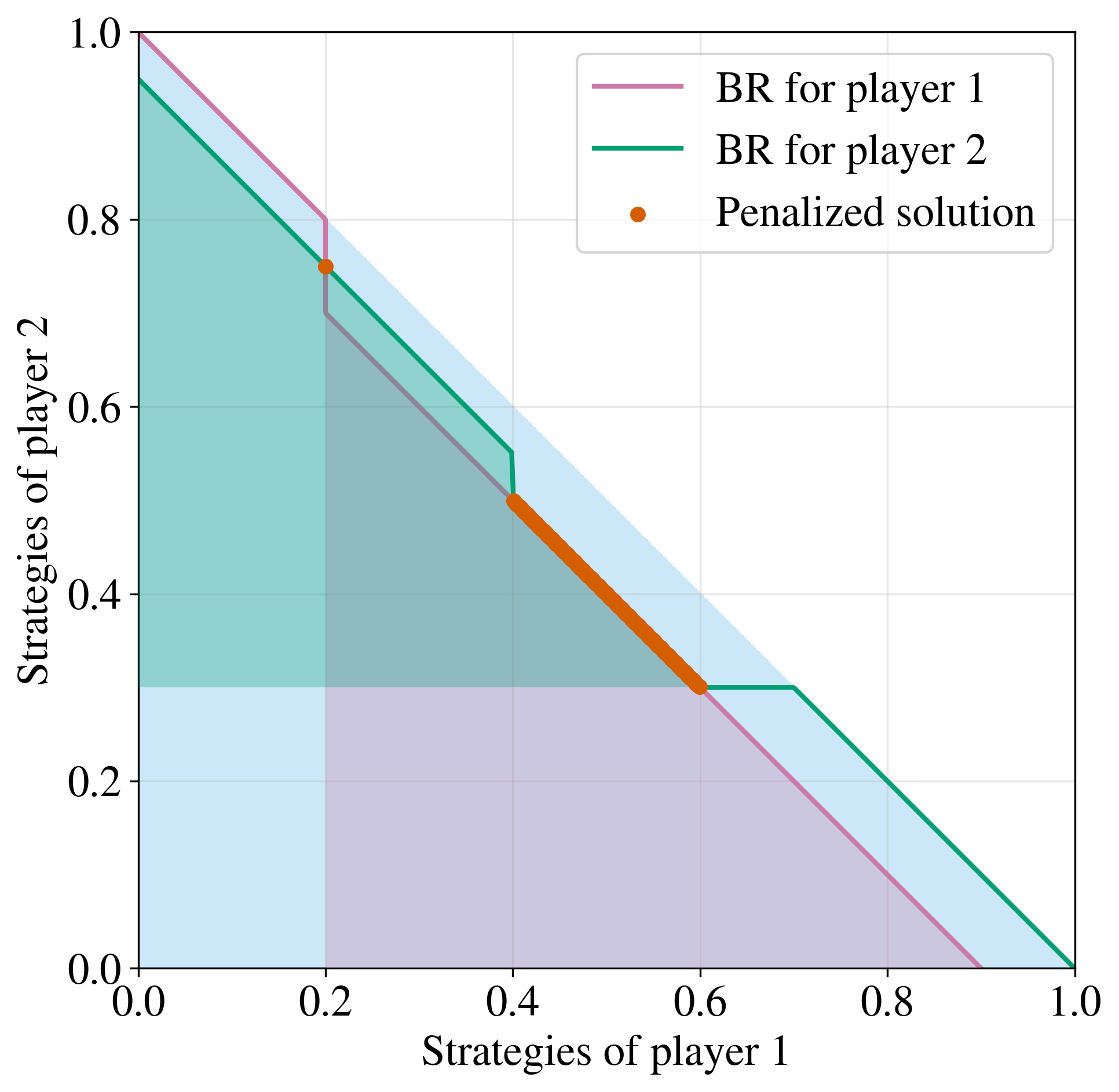}
        (b) with $C_2''$
    \end{minipage}
    \caption{Example~\ref{fig: Solutions_example}a with the modified individual constraint correspondence for player $2$.}
    \label{fig: GNE_counterexamples}
\end{figure}
Finally, the following corollary shows that in case of an exact penalization every GNE of game $G$ is also a $\psi$-penalized solution.
\begin{corollary}
    \label{Cor: GNE_is_penalized}
        Assume that there are constants $\alpha_u, \alpha_\psi > 0$ such that for all players $i \in \mathcal N$ and strategies $a_i \in L_i(a_{-i}^*)$ conditions~\eqref{condition_on_u} and~\eqref{condition_on_psi} hold (see Theorem~\ref{Th: exact_penalization}). 
        If $a^* \in \mathcal A$ is a GNE of game $G$ then $a^*$ is a $\psi$-penalized solution. 
\end{corollary}
\begin{proof}
    We note that $a^*$ is, by definition of a GNE, a solution to Problem~\eqref{Prob: Max_over_min} and thus, by Theorem~\ref{Th: exact_penalization}, a $\psi$-penalized solution.
\end{proof}
Finally, we remark that the conditions stated in Corollary~\ref{Cor: GNE_is_penalized} are sufficient but not necessary. Indeed, in the example presented in Section~\ref{Sec: counterexample} the conditions are not satisfied, but the GNE $(0.5, 0.5)$ is a $\psi$-penalized solution. Nevertheless, without additional conditions, there may be a GNE that is not a $\psi$-penalized solution of game $G$, e.g. the GNE $(0,0)$ in the example in Section~\ref{Sec: counterexample}.

\subsection{Nash bargaining problem}
\label{sec: Nash_bargaining}
Nash considered the \textit{bargaining problem} and proposed a solution in \cite{Nash1950bargaining, Nash1953demand}. The game models a negotiation between two players in the following way. 
First, each player chooses a threat. Second, knowing the threats, each chooses her demand of utility. If there is a strategy profile offering both players their respective demands, they cooperate and play the corresponding strategy profile. Otherwise, they play their threats. 
To solve this game and select an equilibrium, Nash introduced a sequence of games defined by a smoothing function and studied the limit of corresponding equilibria, but did not provide an explicit example of smoothing function. 
van Damme specified this approach in \cite{vanDamme1991} modeling the bargaining problem as two player game $B = (\{1,2\}, (\mathbb R_{\geq 0})_{i \in \{1,2\}}, (v_i)_{i \in \{1,2\}}, \mathcal S)$ where $\mathcal S \subset \mathbb R^2$ denotes the compact and convex set of payoffs that players can obtain by cooperating and the utility $v_i: \mathbb R_{\geq0}^2 \rightarrow \mathbb R$ is given by
\begin{equation}\label{eq;def_utility_unbounded_bargaining}
    v_i(a_i, a_{-i}) = 
    \begin{cases}
        a_i &a \in \mathcal S \\
        0 &\text{otherwise}
    \end{cases}
\end{equation}
Furthermore, he considered a set of perturbations $H = \bigcup_{\epsilon >0} H^\epsilon$ where $H^\epsilon$ is a set of functions such that
\begin{align*}
    H^{\epsilon}
    =
    \{
        h: \mathbb R_{\geq 0}^2 \to (0,1] 
        \, | \,& 
        \text{$h$ continuous}, 
        \forall a \in\mathcal S, h(a) = 1 
        \text{ and }
        \\ 
        &\forall a\not\in\mathcal S, 
        d(a, \mathcal S) > \epsilon
        \Rightarrow
        \max\{h(a), a_1 a_2 h(a) \} < \epsilon
    \}
\end{align*}
and defined the \textit{perturbed game} $B(h)= (\{1,2\}, (\mathbb R_{\geq 0})_{i \in \{1,2\}}, (v^h_i)_{i \in \{1,2\}})$ where $v^h_i:\mathbb R_{\geq 0}^2\rightarrow \mathbb R$ is the utility of player $i$ is such that, for any $(a_1,a_2)\in\mathbb R_{\geq 0}^2$
\begin{equation*}
    v^h_i(a_1, a_2) = a_i h(a_1, a_2)
\end{equation*}
Thus, $v^h_i$ describes the product of the utility function and the function $h$ that gradually penalizes strategies that are not in the set $\mathcal S$. 
An \textit{H-essential equilibrium} is defined as (see Definition 7.5.3 in \cite{vanDamme1991})
\begin{definition}
    A Nash equilibrium $a^*$ of $B$ is an \textit{H-essential equilibrium} if associated with every sequence $(h^\epsilon)_{\epsilon \uparrow 0}$ with $h^\epsilon \in H^\epsilon$ there is a sequence $(a^\epsilon)_{\epsilon \uparrow 0}$ such that $a^\epsilon$ is a Nash equilibrium of $B(h^\epsilon)$ and $a^\epsilon \to a^*$ as $\epsilon \to 0$.
\end{definition}
It is shown that the game B has a unique H-essential equilibrium (Theorem 7.5.5 in \cite{vanDamme1991}) that is the solution proposed by Nash in \cite{Nash1953demand} (Corollary 7.5.6 in \cite{vanDamme1991}). The H-essential equilibrium thus gives an equivalent characterization of the solution to the bargaining problem. \\
In the following, we clarify the relation between $H$-essential equilibria and the $\psi$-penalized solutions for games with compact and convex sets of strategies. 
Therefore, we assume that in the bargaining problem the set of strategies of each player is restricted to a compact and convex subset $\mathcal A_1, \mathcal A_2 \subset \mathbb R_{\geq 0}^2$ such that $\mathcal S \subseteq \mathcal A=\mathcal A_1\times \mathcal A_2$. We define the restricted bargaining game $B_c = (\{1,2\}, (\mathcal A_i)_{i \in \{1,2\}}, (v_i^c)_{i \in \{1,2\}}, \mathcal S)$ such that $\mathcal A_i\subset \mathbb R_{\geq 0}$ is  the compact and convex set of strategies of player $i$, $\mathcal S \subseteq \mathcal A$ and $v_i^c: \mathcal A \rightarrow \mathbb R$ such that
\begin{equation*}
    v_i^c(a_i, a_{-i}) = 
    \begin{cases}
        a_i &a \in \mathcal S \\
        0 &\text{otherwise}
    \end{cases}
\end{equation*}
Furthermore, we consider the set $H_c$ of restricted perturbations such that
\begin{align*}
    H_c
    &:= \bigcup_{\epsilon>0} H_c^\epsilon \\
    &:= \bigcup_{\epsilon>0}
    \left\{
        \begin{aligned}
        h:\mathcal A \to (0,1]
        \ \big|\;&
        h \text{ is continuous},\quad
        \forall a\in\mathcal S,\ h(a)=1, \\[-0.25em]
        &
        \forall a\notin\mathcal S,\ 
        d(a,\mathcal S)>\epsilon
        \Rightarrow
        \max\{h(a),\,a_1a_2h(a)\}<\epsilon
        \end{aligned}
    \right\}.
\end{align*}
and define for $h \in H_c$ the function $v_i^{c,h}: \mathcal A \to \mathbb R$ such that for any $(a_1, a_2) \in \mathcal A$
\begin{equation*}
    v_i^{c,h}(a_1, a_2) = a_i h(a_1, a_2)
\end{equation*}
Then, we denote the restricted perturbed game $B_c(h) = (\{1,2\}, (\mathcal A_i)_{i \in \{1,2\}}, (v_i^{c,h})_{i \in \{1,2\}}, \mathcal S)$. The concept of an $H$-essential equilibrium can be also be defined for the bargaining game with compact and convex set of strategies $B_c$
\begin{definition}
     A Nash equilibrium $a^*$ of $B_c$ is an \textit{$H_c$-essential equilibrium} if associated with every sequence $(h^\epsilon)_{\epsilon \uparrow 0}$ with $h^\epsilon \in H_c^\epsilon$ there is a sequence $(a^\epsilon)_{\epsilon \uparrow 0}$ such that $a^\epsilon$ is a Nash equilibrium of $B_c(h^\epsilon)$ and $a^\epsilon \to a^*$ as $\epsilon \to 0$.
\end{definition}
To study the relation between $H_c$-essential equilibria and $\psi$-penalized solutions, we introduce the game $\tilde{B}_c = (\{1,2\}, (\mathcal A_i)_{i \in \{1,2\}}, (\log \circ \, v_i^c)_{i \in \{1,2\}}, \mathcal S)$. This game is equal to $B_c$ except that, for any player  $i\in\mathcal N$, the utility function $v^c_i$ is replaced by $\log\circ~ v^c_i$. The set $\mathcal S$ describes the feasible strategy profiles and induces the individual constraint correspondences $(S_i)_{i \in \{1,2\}}$, as discussed previously in Section \ref{Sec: setting}. We write, by abuse of notation, also $\tilde{B}_c$ for the equivalent game in which the constraint set $\mathcal S$ is represented by the individual constraint correspondences $(S_i)_{i \in \{1,2\}}$. The following theorem shows that an $H_c$-essential equilibrium of game $B_c$ is a $\psi$-penalized solution of $\tilde{B}_c$.
\begin{theorem}
\label{Th: h_essential_penalized_log}
    If $a^* \in \mathcal A$ is an $H_c$-essential equilibrium of game $B_c$, then $a^*$ is a $\psi$-penalized solution of game $\tilde{B_c}$ for all penalty functions $\psi = (\psi_1,\psi_2)$ such that $\psi_1 = \psi_2$. 
\end{theorem}

\begin{proof}
Let $\psi = (\psi_1, \psi_2)$ be a penalty function with $\psi_1 = \psi_2$ and $(\rho_n)_{n \in \mathbb N}$ a sequence of penalty parameters. Moreover, let $a^*$ be an $H_c$-essential equilibrium of $B_c$. To show that $a^*$ is a $\psi$-penalized solution of game $\tilde{B_c}$, we must find a sequence $(a^*_m)_{m \in \mathbb N}$ such that for any $m \in \mathbb N$, $a^*_m$ is a GNE of the $\psi$-penalized games $\tilde{B}_{c, \rho_{n_m}}$ (defined as in Section \ref{sec: penalized_game}) and $a^*_m \to a^*$ as $m \to \infty$. \\
The utility of player $i$ in $\psi$-penalized game $\tilde{B}_{c, \rho_{n_m}}$ is given by
\begin{equation*}
    \log(a_i) - \rho_{n_m} \psi_i(a_1, a_2)
\end{equation*}
By monotonicity of the exponential function, a strategy profile $a^*$ is thus a GNE of $\tilde{B}_{c, \rho_{n_m}}$ if and only if for all players $i \in \{1,2\}$ it holds
\begin{equation*}
    a_i^* \in \argmax_{a_i \in \mathcal A_i} \log(a_i) - \rho_{n_m} \psi_i(a_i, a_{-i}^*) = \argmax_{a_i \in \mathcal A_i} a_i \exp\big( - \rho_{n_m} \psi_i(a_i, a_{-i}^*) \big)
\end{equation*}
We remark that the constraints in $\tilde{B}_{c, \rho_{n_m}}$ are given by the strategy sets $\mathcal A$. The previous characterization is thus equivalent to such of a NE. \\
Let $(\epsilon_m)_{m \in \mathbb N} \subset \mathbb R_{>0}$ be any sequence such that $\epsilon_m \to 0$ as $m \to \infty$. By definition of an $H_c$-essential equilibrium, for any sequence $(h_m)_{m \in \mathbb N}$ with $h_m \in H_c^{\epsilon_m}$ there is a sequence $(a^*_m)_{m \in \mathbb N}$ such that $a^*_m \to a^*$ as $m \to \infty$ and $a^*_m$ is a NE of $B_c(h_m)$, i.e. for any player $i \in \{1,2\}$
\begin{equation*}
    a_{m,i}^* \in  \argmax_{a_i \in \mathcal A_i} a_i h_m(a_i, a_{m,-i}^*)
\end{equation*}
It is thus sufficient to show that there is a $M \in \mathbb N$ such that for all $m \geq M$ it holds $p_m := \exp\big( - \rho_{n_m} \psi_i(a_i, a_{-i}^*) \big) \in H_c^{\epsilon_m}$. \\
First, we remark that $\psi: \mathcal A \to \mathbb R_{\geq 0}$ is assumed to be continuous and thus bounded on the compact set $\mathcal A$. Because $\rho_m \to \infty$ as $m \to \infty$, there is thus a $M \in \mathbb N$ such that for all $m \geq M$ the function $p_m: \mathcal A \to (0,1]$ is well-defined and continuous.\\
Second, by definition of the penalty function $\psi$, we remark that for all $a \in \mathcal S$, $p_m(a) = 1$. \\
It is thus left to show that for all $a \notin \mathcal S$
\begin{equation*}
    d(a, \mathcal S) > \epsilon \Rightarrow \max\{p_m(a), a_1 a_2 p_m(a)\} < \epsilon_m
\end{equation*}
Therefore, we remark that the set $E_m := \{a \in \mathcal A ~|~ d(a,S) \geq \epsilon_m\}$ is compact. Moreover, $\psi$ is supposed to be continuous and strictly positive in points $(a_1, a_2) \in E_m$, we can thus deduce that there is a constant $c_m > 0$ such that for all $a \in E_m$
\begin{equation*}
    c_m \leq \psi(a)
\end{equation*}
Besides, the function $f: \mathcal A \to \mathbb R$ defined by $(a_1, a_2) \to a_1 a_2$ is also continuous on the compact $E_m$ and thus bounded. Hence, there is another constant $d_m > 0$ such that for all $a \in E_m$
\begin{equation*}
    a_1 a_2 \leq d_m
\end{equation*}
We deduce that there is a $M_m \in \mathbb N$ such that for all $n \geq M_m$ and $a \in \mathcal A$ with $d(a, \mathcal S) > \epsilon_m$
\begin{equation*}
        \max\big(\exp(-\rho_{n_m} \psi(a)), a_1 a_2 \exp(-\rho_{n_m} \psi(a)) \big) < \epsilon_m
\end{equation*}

We conclude that $p_m \in H_c^{\epsilon_m}$. 
\end{proof}

Moreover, in the following theorem we use the invariance of the set of $\psi$-penalized solutions under composition of the penalty or utility with monotonously increasing functions to show that $a^*$ is also a $\psi$-penalized solution of $B_c$. With a slight abuse of notation, we denote again by $B_c$ the game with individual constraint correspondences $(S_i)_{i \in \{1,2\}}$ instead of the constraint set $\mathcal S$.
\begin{theorem}
\label{Th: h_essential_is_penalized}
    If for both players $i \in \{1,2\}$ the following conditions hold
    \begin{itemize}
        \item[(i)] for all $a_{-i} \in \mathcal A_{-i}$, $S_i$ is l.s.c. in $a_{-i}$ or $S_i(a_{-i})$ is a singleton,
        \item[(ii)] for all $a_{-i} \in \mathcal A_{-i}$, $\psi(\cdot, a_{-i})$ is strictly convex on $\mathcal A_i \setminus S_i(a_{-i})$,
    \end{itemize}
    and $a^* \in \mathcal A$ is an $H_c$-essential equilibrium of $B_c$ such that 
    \begin{itemize}
        \item[(iii)] there is a constant $\alpha_\psi > 0$ such that for all players $i = 1, 2$ and strategies $a_i \in \mathcal A_i$,
        \begin{equation*}
            \begin{aligned}
                | \psi_i(a_i, a_{-i}^*) - \psi_i(a^*) | \geq \alpha_{\psi} |a_i - a_i^*| 
                \quad \text{or} \quad \psi(a^*) = \psi(a_i, a_{-i}^*)
            \end{aligned}
        \end{equation*}
    \end{itemize}
    Then $a^*$ is a $\psi$-penalized equilibrium of $B_c$ for all penalty functions $\psi = (\psi_1,\psi_2)$ such that $\psi_1 = \psi_2$. 
\end{theorem}
\begin{proof}
    By Theorem \ref{Th: h_essential_penalized_log}, an $H_c$-essential equilibrium $a^*$ is a $\psi$-penalized solution of the game $\tilde B_c$. \\
    In the following, we apply Theorems~\ref{Th: penalized_solves_problem} and~\ref{Th: exact_penalization}. Hereby, the coupled constraint $\mathcal L$ is defined as the strategy sets $\mathcal A$. Consequently, the coupled constraint correspondences are taking the constant value $\mathcal A_i$ and are thus u.s.c. and l.s.c. The individual constraint correspondences $S_i$ are assumed to satisfy Assumption (i). Moreover, the utility $\log \circ~ v_i$ is concave and $\psi$ assumed to satisfy Assumption (ii). We apply thus Theorem~\ref{Th: penalized_solves_problem} and deduce that, for every player $i \in \{1,2\}$, the strategy $a_i^*$ solves
    \begin{equation*}
        \begin{aligned}
            &\max_{a_i} \log(v_i(a_i, a_{-i}^*)) \\
            &\text{s.t. } a_i \in \argmin_{b_i \in S_i(a_{-i}^*)} \psi_i(b_i, a_{-i}^*)
        \end{aligned}
    \end{equation*}
    Then, by the monotonicity of the exponential function, it solves also the following problem
    \begin{equation*}
       \begin{aligned}
           &\max_{a_i} v_i(a_i, a_{-i}^*) \\
            &\text{s.t. } a_i \in \argmin_{b_i \in S_i(a_{-i}^*)} \psi_i(b_i, a_{-i}^*)
       \end{aligned}
    \end{equation*}
    Besides, for all players $i =1,2$ and strategies $a_{-i}$ the utility $v_i(\cdot, a_{-i})$ is $1$-Lipschitz continuous. Finally, we deduce, by Theorem~\ref{Th: exact_penalization}, that $a^*$ is a $\psi$-penalized solution of $B_c$.
\end{proof}
We note that a $\psi$-penalized equilibrium is not necessarily an $H_c$-essential equilibrium. To see that, it is sufficient to consider $\mathcal A_1 = \mathcal A_2 = [0,2]$, $\mathcal S = \{ (a_1, a_2) \in [0,2]^2 ~|~ a_1 + a_2 \leq 1 \}$ and $\psi$ chosen as in \eqref{eq:penalty_distance}. This game has infinitely many GNEs, i.e. all strategy profiles $(a_1, a_2) \in [0,2]^2$ with $a_1 + a_2 = 1$, and all of them are also $\psi$-penalized solutions (Corollary \ref{Cor: GNE_is_penalized}). However, by adapting the proof of Theorem 7.5.5 in \cite{vanDamme1991} to this game, we know that the $H_c$-essential equilibrium is unique and there must thus be infinitely many $\psi$-penalized solutions that are not a $H_c$-essential equilibrium. \\
Finally, by adapting the proof of Theorem 7.5.4 in \cite{vanDamme1991} to compact and convex strategy sets $\mathcal A_1$ and $\mathcal A_2$, it can be shown that the solution of the bargaining problem is an $H_c$-essential equilibrium of $B$ and thus a $\psi$-penalized solution of $\tilde B_c$ for all penalty functions $\psi$ (Theorem \ref{Th: h_essential_penalized_log}). Under additional conditions on the regularity of $\mathcal S$ and on the structure of $\psi$, the solution of the bargaining problem is a $\psi$-penalized solution of $B_c$ (Theorem \ref{Th: h_essential_is_penalized}). \\
To conclude, it remains an open question how the $H_c$-essential equilibrium can be extended to general games and how it relates to the $\psi$-penalized solution in such settings.

%% file: Conclusion_short.tex
\section{Future Works}
\label{Sec: Future_works}
In the future, we aim to generalize our main results and find weaker conditions under which solutions to problem~\eqref{Prob: Max_over_min} are $\psi$-penalized solutions, and vice versa. In particular, it remains an open question to determine conditions such that strategy profiles solving problem~\eqref{Prob: Max_over_min} for all players are $\psi$-penalized solutions without implying an exact penalization, or allowing also differentiable penalty functions (as in Section~\ref{Sec: sufficient_and_exact_penalization}). 
Furthermore, the $\psi$-penalized solution depends on the choice of $\psi$ and characterizing the solutions (e.g., topological structure, robustness, fairness) with respect to penalty function is thus an important area for further research.
Besides, future works should study the conditions for uniqueness of the $\psi$-penalized solution, further refinements and relations to other established solution concepts. Finally, generalizing the solution concept to games with discrete strategy sets or general metric spaces is also an interesting direction that may allow for more applications.

%% file: Appendix_short.tex
\section*{Appendix}
\label{Sec: Appendix}
In this work, we use the notations in Table \ref{tab: notations}.
\input{Notations}
\\
First, we recall the definition of l.s.c. (see p.109 in \cite{berge1963}).
\begin{definition}
\label{Def: lsc}
  Let $C: X \to Y$ be a correspondence between topological spaces $X$ and $Y$. Let $x_0$ be a point in $X$. $C$ is said to be \textit{lower semi-continuous} (l.s.c.) at $x_0$ if for each open set $U \subset Y$ with $C(x_0) \cap U \neq \emptyset$ there is a neighborhood $N(x_0) \subset X$ such that for all $x \in N(x_0)$ it holds $C(x) \cap U \neq \emptyset$.
  We say that $C$ is l.s.c. in $X$ if it is l.s.c. at $x$ for all $x \in X$.
\end{definition}
We remark that by Theorem 1 on p.109 in \cite{berge1963} and Corollary 1.1 in \cite{hogan1973}, the following equivalence holds.
\begin{lemma}
\label{Lem: Hogan}
    A correspondence $C: X \to Y$ between topological spaces $X$ and $Y$ is l.s.c. at $x_0 \in X$ if and only if for all sequences $(x_n)_{n \in \mathbb N}$ converging to $x_0$ and $y_0 \in C(x_0)$ there is $N \in \mathbb N$ and a sequence $(y_n)_{n \in \mathbb N}$ converging to $y_0$ such that $y_n \in C(x_n)$ for all $n \geq N$. 
\end{lemma}
Now, we prove a preliminary result on l.s.c. correspondences. 
\begin{lemma}
\label{Lem: compactification_is_lsc}
    Let $C: X \to 2^Y$ be a correspondence between two topological spaces $X$ and $Y$. If $C$ is l.s.c. in $x_0$, the correspondence $\overline{C}: X \to Y$ defined by $\overline{C}(x) = \overline{C(x)}$ is also l.s.c. in $x_0$.
\end{lemma}
\begin{proof}
    Let $G$ be an open set such that $\overline{C}(x_0) \cap G \neq \emptyset$. Then there is some $x \in \overline{C}(x_0) \cap G$. Since $x \in G$ and $G$ open, $G$ is an open neighborhood of $x$. By the definition of the closure of a set, we have $C(x_0) \cap G \neq \emptyset$ (see p.5-6 in \cite{steen1970counterexamples} and Theorem 3 in \cite{berge1963}). Furthermore, since $C$ is l.s.c. in $x_0$, there is a neighborhood $U(x_0)$ such that for all $x \in U(x_0)$, it holds $C(x) \cap G \neq \emptyset$ and thus also $\overline{C}(x) \cap G \neq \emptyset$. We conclude that $\overline{C}$ is l.s.c. in $x_0$.
\end{proof}

%% file: Notations.tex
\begin{table}[ht]
	\begin{minipage}{\columnwidth}
		\begin{center}
			\begin{tabular}{ll}
				\toprule
                    $G$ & original game $G = \big(\mathcal{N}, (\mathcal{A}_i)_{i \in \mathcal{N}} , (u_i)_{i \in \mathcal{N}}, (C_i)_{i \in \mathcal N} \big)$  \\
                    $\mathcal{N}$ & set of players \\
                    $N$ & cardinality of set of players \\
                    $\mathcal A_i$ & set of strategies of player $i \in \mathcal N$ \\
                    $\mathcal A$ & product $\prod_{i \in \mathcal N} \mathcal A_i$ \\
                    $\mathcal A_{-i}$ & product $\prod_{i \in \mathcal N, i \neq j} \mathcal A_j$ \\
                    $u_i$ & utility of player $i \in \mathcal N$ \\
                    $C_i$ & individual constraint correspondence of player $i \in \mathcal N$ \\
                    $\mathcal L$ & coupled constraint \\
                    $L_i$ & coupled constraint correspondence of player $i \in \mathcal N$\\
                    $G_\rho$ & $\psi$-penalized game $G_\rho = \big(\mathcal{N}, (\mathcal{A}_i)_{i \in \mathcal{N}} , (u_i - \rho \psi_i)_{i \in \mathcal{N}}, $\\
                    &$(L_i)_{i \in \mathcal N} \big)$  \\
                    $\rho$ & penalty parameter \\
                    $\psi_i$ & penalty function of player $i \in \mathcal N$\\
                    $G^{ce}$ & game serving as a counterexample in Section \ref{Sec: counterexample}\\
                    $G^{DD}$ & game serving as an example in Section \ref{Sec: DD_with_constraints}\\
                    $d(x, A)$ & distance of a point $x$ to a set $A$ s.t. $d(x,A) := \inf_{a \in A} |x - a|$ \\
                    $\mathcal P(A)$ & power set of a set $A$
                    \\
                    $\Gamma(f)$ & graph of a function or correspondence$f$\\
                    $\bar{A}$ & closure of a set $A$ \\
				\bottomrule
			\end{tabular}
		\end{center}
		\bigskip\centering
	\end{minipage}
    \caption{Notations}
    \label{tab: notations}
\end{table}